\documentclass[12pt, reqno]{amsart}
\usepackage{amsmath, amsthm, amscd, amsfonts, amssymb, graphicx, color}
\usepackage[bookmarksnumbered, colorlinks, plainpages]{hyperref}
\newtheorem{theorem}{Theorem}[section]
\newtheorem{lemma}[theorem]{Lemma}
\newtheorem{proposition}[theorem]{Proposition}

\theoremstyle{definition}

\newtheorem{example}[theorem]{Example}

\theoremstyle{remark}
\newtheorem{remark}[theorem]{Remark}
\numberwithin{equation}{section}

\begin{document}

\title[Convexity of the  Csisz\'{a}r $f$-divergence functional]{How type of Convexity of the Core function affects the Csisz\'{a}r $f$-divergence functional}
\author[M. Kian]{Mohsen Kian}

\address{Mohsen Kian: Department of Mathematics, University of Bojnord, P. O. Box 1339, Bojnord 94531, Iran}
\email{kian@ub.ac.ir }



\begin{abstract}
We investigate how the type of Convexity of the Core function affects the Csisz\'{a}r $f$-divergence functional. A general treatment for the type of convexity has been considered and the associated perspective functions have been studied. In particular, it has been shown that when the core function is \rm{MN}-convex, then the   associated perspective function is jointly \rm{MN}-convex if the two scalar mean \rm{M} and \rm{N} are the same.    In the case where $\mathrm{M}\neq\mathrm{N}$, we study the type of convexity of the perspective function.
As an application, we prove that the   \textit{Hellinger distance}   is jointly \rm{GG}-convex.
As further applications,  the matrix Jensen inequality has been developed for the perspective functions under different kinds of convexity.
\end{abstract}
\maketitle

%%%%%%%%%%%%%%%%%%%%%%%%%%%%%%%%%%%%%%%%%%%%%%%%%%%%%%%%%%%%%%%%%%%%
\section{introduction}

In the probability theory, the notion of Csisz\'{a}r $f$-divergence is well-known in relation with measures between probability distributions. Those kinds of measures have many applications in many directions, like economics, genetics, signal processing and so on. In fact,  Csisz\'{a}r \cite{CS,CS2} introduced   $f$-divergence functional of a function $f:[0,\infty)\to\mathbb{R}$ by
\begin{align*}
  I_f(\mathbf{p},\mathbf{q}):=\sum_{j=1}^{n}q_j f\left(\frac{p_j}{q_j}\right)
\end{align*}
 for  $n$-tuples of positive real numbers $\mathbf{p}=(p_1,\ldots,p_n)$ and $\mathbf{q}=(q_1,\ldots,q_n)$. In above definition, the undefined
expressions are  interpreted as
\begin{eqnarray*}
 f(0)=\lim_{t\to 0^+}f(t), \qquad 0f\left(\frac{0}{0}\right)=0, \qquad
 0f\left(\frac{p}{0}\right)=\lim_{\epsilon\to 0^+}f\left(\frac{p}{\epsilon}\right)=p\lim_{t\to \infty}\frac{f(t)}{t}.
 \end{eqnarray*}

 A useful fact concerning the $f$-divergence functional was proved by Csisz\'{a}r and K\"{o}rner \cite{CK} as follows. In fact, they showed that the perspective function of a convex function is sub-additive. \par

\textbf{Theorem A.} \textit{If $f:[0,\infty)\to\mathbb{R}$ is convex, then $I_f(\mathbf{p},\mathbf{q})$ is jointly convex in $\mathbf{p}$ and $\mathbf{q}$ and
 \begin{eqnarray}\label{m1}
 g\left(\sum_{j=1}^n p_j,\sum_{j=1}^n q_j\right)=\sum_{j=1}^{n}q_jf\left(\frac{\sum_{j=1}^{n}p_j}{\sum_{j=1}^{n}q_j}\right)\leq I_f(\mathbf{p},\mathbf{q})
 \end{eqnarray}
 for all positive $n$-tuples $\mathbf{p}=(p_1,\ldots,p_n)$ and $\mathbf{q}=(q_1,\ldots,q_n)$, where the  perspective function $g$ associated to $f$ is defined   by
\begin{align*}
 g(x,y):=yf\left(\frac{x}{y}\right).
\end{align*}}

 When   $f$ varies through convex functions, the Csisz\'{a}r $f$-divergence produces different known measures. Among others, we mention the following notable measures:\\
-\textit{Kullback--Leibler distance} is defined by $KL(\mathbf{p},\mathbf{q}):=\sum_{j=1}^{n}p_j \log\left(\frac{p_j}{q_j}\right)$ and $KL=I_f$, when     $f(t)=t\ln t$ $(t>0)$.\\
-\textit{Total variation distance} is defined by $V(\mathbf{p},\mathbf{q}):=\sum_{j=1}^{n}\left|p_j-q_j\right|$ and $V=I_f$, when     $f(t)=|t-1|$ $(t\geq0)$.\\
-\textit{ Hellinger distance} is defined by $H^2(\mathbf{p},\mathbf{q}):=2\sum_{j=1}^{n}\left(\sqrt{p_j}-\sqrt{q_j}\right)^2$ and $H^2=I_f$, when     $f(t)=2(\sqrt{t}-1)^2$\, $(t\geq0)$.\\
-\textit{ $\chi^2$-distance} is defined by $D_{\chi^2}(\mathbf{p},\mathbf{q}):=\sum_{j=1}^{n}\frac{\left(p_j-q_j\right)^2}{q_j}$ and $D_{\chi^2}=I_f$, when     $f(t)= (t-1)^2$\, $(t\geq0)$.\\
-\textit{R\'{e}nyi’s divergences} are defined by $R_\alpha(\mathbf{p},\mathbf{q}):=\frac{1}{\alpha (\alpha -1)}\ln \rho_\alpha(\mathbf{p},\mathbf{q})$ for every $\alpha\in\mathbb{R}\backslash\{0,1\}$, where $\rho_\alpha(\mathbf{p},\mathbf{q})= \sum_{j=1}^{n}p_j^\alpha q_j^{1-\alpha}$ and $\rho_\alpha=I_f$, when     $f(t)=t^\alpha$\, $(t>0)$.\\

For more information about $f$-divergence functional and its properties, the reader is referred to \cite{Gil,H-M,Ki,Sa,Va} and references therein.

  For every two positive real numbers $x,y$ and every $\alpha\in[0,1]$,  the most  well-known scalar means read as follows:
 \begin{align*}
   \mathrm{\mathbf{A}}_\alpha(x,y)&=\alpha x+(1-\alpha)y \qquad\qquad\quad \mbox{Arithmetic mean}\\
 \mathrm{\mathbf{G}}_\alpha(x,y)&=x^\alpha y^{1-\alpha}\qquad\qquad\qquad\qquad \mbox{Geometric  mean}\\
\mathrm{\mathbf{H}}_\alpha(x,y)&= \left(\alpha x^{-1}+(1-\alpha)y^{-1}\right)^{-1}\quad \mbox{Harmonic  mean}.
 \end{align*}
 The Arithmetic-Geometric-Harmonic means inequality is well-known:
\begin{align}\label{agh}
\mathrm{\mathbf{H}}_\alpha(x,y)\leq  \mathrm{\mathbf{G}}_\alpha(x,y) \leq   \mathrm{\mathbf{A}}_\alpha(x,y),\qquad        (x,y\geq0,\ \alpha\in[0,1]).
\end{align}

\section{The effect of type of convexity of  core function on the   $f$-divergence functional}

Convex functions are known to be defined using the Arithmetic mean: A real function $f$ is convex  when
 $$f(\mathrm{\mathbf{A}}_\alpha(x,y))\leq \mathrm{\mathbf{A}}_\alpha(f(x),f(y))$$
for all $x,y$ in domain of $f$ and every $\alpha\in[0,1]$. However, when
 the Arithmetic means are replaced by other means in both sides of the above inequality, different type of convexities for functions can be derived.
In the next definition, we limit the domain and the range of our function to the positive half-line, while it will be  possible to consider this sets a more general subset of  real functions depending on occasions. If $\mathrm{\mathbf{M}}_\alpha$ and $\mathrm{\mathbf{N}}_\alpha$ are two $\alpha$-weighted scalar means, a positive real function $f$ on $(0,\infty)$ is said to be  \rm{MN}-convex, when
\begin{align}\label{PQ}
  f(\mathrm{\mathbf{M}}_\alpha(x,y))\leq \mathrm{\mathbf{N}}_\alpha(f(x),f(y))
\end{align}
holds for all $x,y\geq0$   and every $\alpha\in[0,1]$. Note that an \rm{AA}-convex function is simply called convex. Moreover, some of these  functions enjoy  well-known titles. \rm{HA}-convex functions are called Harmonically convex and \rm{AG}-convex functions are known as log-convex or multiplicative convex functions.

Some of basic facts concerning \rm{MN}-convex functions are given in the following lemma. The reader is referred to \cite{An-Va-Vu,Dr1,Dr3,Ni,No,Zh-Cu-Zh} to see the proofs and more information about these functions.
 \begin{lemma}\label{all-con}
Let $f$ be a positive real function on $(0,\infty)$.\\
 \rm{(i)}\ If $f$ is \rm{AG}-convex if and only if $\log f$ is convex;\\
\rm{(ii)}\   $f$ is \rm{AH}-convex  if and only if $1/f$ is concave;\\
\rm{(iii)}\  $f$ is \rm{GA}-convex (concave)   if and only if $fo\exp$  is convex (concave);\\
\rm{(iv)}\  If $h$ is convex (concave), then $f(t)=h(\ln t)$ is  \rm{GA}-convex (concave);\\
\rm{(v)}\ $f$ is \rm{GG}-convex if and only if  the function $h=\ln ofo\exp$  is convex;\\
\rm{(vi)}\ $f$ is \rm{GG}-convex if and only if $h=\ln o f$ is \rm{GA}-convex;\\
\rm{(vii)}\ $f$ is \rm{GH}-convex (concave)   if and only if $fo\exp$ is \rm{AH}-convex (concave);\\
\rm{(iix)}\ $f$ is \rm{HG}-convex  if and only if $h(t)=t\ln f(t)$ is convex;\\
\rm{(ix)} $f$ is \rm{HG}-convex if and only if  $\ln f$ is \rm{HA}-convex;\\
\rm{(x)}    $f$ is \rm{HH}-convex (concave)   if and only if  $h(t)=t/f(t)$  is  concave (convex).
\end{lemma}

We remarked that each class of \rm{MN}-convex functions we mentioned in Lemma \ref{all-con} actually contains many examples. So we give several examples  before we continue.
\begin{example}
The functions $t\mapsto 1/\sqrt{t}$ and $t\mapsto -t^{-3}$  are \rm{AH}-convex on $(0,\infty)$. \\
The functions $t\mapsto \exp t$ and $t\mapsto t^r$ $(r<0)$ are \rm{AG}-convex on $\mathbb{R}$ and $(0,\infty)$, respectively.\\
  The function $t\mapsto \log (1+t)$ is \rm{GA}-convex on $(0,\infty)$. Moreover, recall that the  well-known digamma function is defined by $\psi(t)=\frac{d}{dt}\log \Gamma(t)=\frac{\Gamma(t)}{\Gamma'(t)}$ on $(0,\infty)$, where $\Gamma$ denotes the gamma function, i.e.,
$\Gamma(t)=\int_{0}^{\infty}x^{t-1}e^{-x}dx$. It is known that \cite{Zh-Cu-Zh} the functions $t\mapsto \psi(t)+\frac{1}{2t}$ and $t\mapsto \psi(t)+\frac{1}{2t}+\frac{1}{12t^2}$ are \rm{GA}-concave  and \rm{GA}-convex, respectively,  on $(0,\infty)$.\\
It has been shown in \cite{Ni} that if $f(t)=\sum_{n=0}^{\infty}c_n t^n$ is a real analytic function whose radius of convergence is  $r>0$ and whose coefficients $c_n$ are non-negative, then
 $f$ is a \rm{GG}-convex function on $(0,r)$. This implies that the functions $\exp$, $\sinh$ and $\cosh$ are \rm{GG}-convex on $\mathbb{R}$ and the functions $\sec$, $\csc$ and $\tan$ are \rm{GG}-convex on $(0,\pi/2)$. In addition, the functions $t\mapsto(1-t)^{-1}$ and $t\mapsto\frac{1+t}{1-t}$ are \rm{GG}-convex on $(0,1)$, see \cite{Drr}. \\
The functions $t\mapsto\frac{1}{\sqrt{\ln t}}$ and $t\mapsto -(\ln t)^{-3}$ are \rm{GH}-convex on $(0,\infty)$. \\
For all $r\geq 0$ and $r\leq -1$, the function  $t\mapsto\exp(t^r)$   are \rm{HG}-convex on $(0,\infty)$. \\
The functions $t\mapsto\frac{t}{\ln t}$ and $t\mapsto t^r$ $(0\leq r\leq 1)$ are \rm{HH}-convex on $(0,\infty)$.
\end{example}
%%%%%%%=======================================================================================
Regarding the Jensen inequality, Lemma \ref{all-con} can be used to demonstrate variants of the Jensen inequality for every  \rm{MN}-convex function. The proof of next lemma   easily follows from Lemma   \ref{all-con} and the classical Jensen inequality for convex functions. So we do not include details.
\begin{lemma}\label{j-all}
  Let $f$ be a non-negative real function on $(0,\infty)$ and for $i=1,\ldots,n$  let $x_i\geq0$ and $\alpha_i\in[0,1]$ with $\sum_{i=1}^{n}\alpha_i=1$. \\
 \rm{(i)}\ If $f$ is \rm{AG}-convex, then
\begin{align}\label{jag}
f\left(\sum_{i=1}^{n}\alpha_i x_i\right)\leq \prod_{i=1}^{n} f(x_i)^{\alpha_i}.
\end{align}
\rm{(ii)}\  If $f$ is \rm{AH}-convex, then
\begin{align}\label{jah}
f\left(\sum_{i=1}^{n}\alpha_ix_i\right)\leq\left(\sum_{i=1}^{n}\frac{\alpha_i}{f(x_i)}\right)^{-1}.
\end{align}
\rm{(iii)}\ If  $f$ is \rm{GA}-convex, then
\begin{align}\label{jga}
f\left(\prod_{i=1}^{n} x_i^{\alpha_i}\right)\leq \sum_{i=1}^{n}\alpha_if(x_i).
\end{align}
\rm{(iv)}\ If $f$ is \rm{GG}-convex, then
\begin{align}\label{jgg}
f\left(\prod_{i=1}^{n} x_i^{\alpha_i}\right)\leq\prod_{i=1}^{n} f(x_i)^{\alpha_i}.
\end{align}
\rm{(vi)}\ If $f$ is \rm{GH}-convex, then
\begin{align}\label{jgh}
f\left(\prod_{i=1}^{n} x_i^{\alpha_i}\right)\leq\left(\sum_{i=1}^{n}\frac{\alpha_i}{f(x_i)}\right)^{-1}.
\end{align}
\rm{(vii)}\ If $f$ is \rm{HG}-convex, then
 \begin{align}\label{jhg}
f\left(\left(\sum_{i=1}^{n}\frac{\alpha_i}{x_i}\right)^{-1}\right)\leq\prod_{i=1}^{n} f(x_i)^{\alpha_i}.
\end{align}
\rm{(iix)}   If $f$ is \rm{HH}-convex, then
 \begin{align}\label{jhh}
f\left(\left(\sum_{i=1}^{n}\frac{\alpha_i}{x_i}\right)^{-1}\right)\leq
\left(\sum_{i=1}^{n}\frac{\alpha_i}{f(x_i)}\right)^{-1}.
\end{align}
\end{lemma}

%%%@@@@@@@@@@@@@@@@@@@@@@@@@@@@@@@@@@@@@@@@@@@@@@@@@@@@@@@@@@@@@@@@@@@@@@@@@@
%%%%%%%%%%%%%%%%%%%%%%%%%%%%%%%%%%%%%%%%%%%%%%%%%%%%%%%%%%%%%%%%%%%%%%%%%%%%%
\bigskip

We begin with modifications of the celebrated result of  Csisz\'{a}r, \rm{Theorem A}. A  consequence of Theorem A  is that if $f$ is convex,   the associated  perspective function $g_f$ is  convex in both variables. In the next theorem, we investigate the affect of type of convexity of the generating function $f$ to the convexity of the associated perspective function $g_f$. When there is not fear of ambiguity, we briefly use $g$ for the associated perspective function of $f$. Once more, we note that although we restrict the domain and the range of our function to the positive half-line, depending on situation,   it is  possible to consider this sets   more general subsets of  real functions.

\begin{theorem}\label{main1}
Let  $f:(0,\infty)\to(0,\infty)$ be a real function.\\
{\rm (i)}\  $f$ is   \rm{AH}-convex if and only if     $g$ is \rm{AH}-convex on the first coordinate and convex on the second coordinate. In particular, the inequality
\begin{align}\label{tm}
g\big(\mathrm{\mathbf{A}}_\alpha(a,b),\mathrm{\mathbf{A}}_\alpha(x,y) \big)\leq  \mathrm{\mathbf{H}}_\alpha\left\{\left[\mathrm{\mathbf{A}}_\alpha(g(a,x),g(a,y))\right],\left[
\mathrm{\mathbf{A}}_\alpha(g(b,x),g(b,y))
\right]\right\}
\end{align}
holds  for all $a,b,x,y\geq0$ and every $\alpha\in[0,1]$.\\
%%-------------------------------------------------------------------------------
{\rm (ii)}   $f$ is   \rm{AG}-convex if and only if     $g$ is \rm{AG}-convex on the first coordinate and convex on the second coordinate. In particular, the inequality
\begin{align}\label{qaa}
g\big(\mathrm{\mathbf{A}}_\alpha(a,b),\mathrm{\mathbf{A}}_\alpha(x,y) \big)\leq
\mathrm{\mathbf{G}}_\alpha\left\{ \mathrm{\mathbf{A}}_\alpha(g(a,x),g(a,y)),\mathrm{\mathbf{A}}_\alpha(g(b,x),g(b,y))\right\}
\end{align}
holds, for all $a,b,x,y\geq0$ and every $\alpha\in[0,1]$.\\
%%-------------------------------------------------------------------------------
{\rm (iii)}   $f$ is   \rm{GG}-convex if and only if    $g$ is jointly \rm{GG}-convex. In particular, \begin{align}\label{qbb}
g\big(\mathrm{\mathbf{G}}_\alpha(a,b),\mathrm{\mathbf{G}}_\alpha(x,y) \big)\leq
\mathrm{\mathbf{G}}_\alpha\left\{ g(a,x),g(b,y))\right\}
\end{align}
 for all $a,b,x,y\geq0$ and every $\alpha\in[0,1]$.\\
%%-------------------------------------------------------------------------------
{\rm (iv)}   $f$ is   \rm{HH}-convex if and only if    $g$ is jointly \rm{HH}-convex. In particular, \begin{align}\label{qc}
g\big(\mathrm{\mathbf{H}}_\alpha(a,b),\mathrm{\mathbf{H}}_\alpha(x,y) \big)\leq
\mathrm{\mathbf{H}}_\alpha\left\{ g(a,x),g(b,y))\right\}
\end{align}
 for all $a,b,x,y\geq0$ and every $\alpha\in[0,1]$.\\
%%-------------------------------------------------------------------------------
{\rm (v)}   $f$ is   \rm{GH}-convex if and only if     $g$ is \rm{GH}-convex in its first variable and \rm{GG}-convex in its second variable. In particular, the inequality
\begin{align}\label{qv}
g\big(\mathrm{\mathbf{G}}_\alpha(a,b),\mathrm{\mathbf{G}}_\alpha(x,y)\big)\leq
\mathrm{\mathbf{H}}_\alpha\left\{\mathrm{\mathbf{G}}_\alpha(g(a,x),g(a,y)),
\mathrm{\mathbf{G}}_\alpha(g(b,x),g(b,y))\right\}
\end{align}
holds   for all $a,b,x,y\geq0$ and every $\alpha\in[0,1]$. Moreover, in this case $g$ is jointly  \rm{GG}-convex, i.e., \eqref{qbb} holds.
\end{theorem}

Before   proving  Theorem \ref{main1}, we would like to note that if $f$ is \rm{MN}-convex, then $g$ is not necessarily  \rm{MN}-convex in both variables, unless $\mathrm{M}=\mathrm{N}$. For example, if $f$ is  \rm{AH}-convex, then part \rm{(i)} of  Theorem \ref{main1} shows that $g$ is    \rm{AH}-convex on the first coordinate and convex on the second coordinate. However, $g$ is not \rm{AH}-convex in both variables. To see this, consider the \rm{AH}-convex function $f(t)=1/\sqrt{t}$ and put $\alpha=1/2$, $a=1$, $x=2$ and  $y=4$. Then
$$3\sqrt{3}=g(a,\mathrm{\mathbf{A}}_{1/2}(x,y))\nleq \mathrm{\mathbf{H}}_{1/2}(g(a,x),g(a,y))=\frac{16\sqrt{2}}{4+\sqrt{2}}.$$
Note in addition that when  $f$ is  \rm{AH}-convex, $g$ is \rm{AA}-convex in both variables. However, the the reverse direction does not hold, i.e., if $g$ is   \rm{AA}-convex in both variables, then $f$ is not necessarily  \rm{AH}-convex.

\begin{proof}\textit{of Theorem \ref{main1}.}
First assume that $f$ is \rm{AH}-convex. For all $x,y,a\geq0$ and every $\alpha\in[0,1]$ we have
\begin{align*}
 g\big(\alpha a +(1-\alpha) b,x\big)&=xf\Big(\frac{1}{x}\big(\alpha a+(1-\alpha)b\big)\Big)\\
&\leq x\Big[\alpha f\big(\frac{a}{x}\big)^{-1}+(1-\alpha)f\big(\frac{b}{x}\big)^{-1}\Big]^{-1}\notag\\
&=\Big[\alpha x^{-1} f\big(\frac{a}{x}\big)^{-1}+(1-\alpha) x ^{-1}f\big(\frac{b}{x}\big)^{-1}\Big]^{-1}\\
&=\Big[\alpha g(a,x)^{-1}+(1-\alpha) g(b,x)^{-1}\Big]^{-1}.\notag
\end{align*}
This ensures that $g$ is  AH-convex on the first coordinate.   Therefore
\begin{align*}
g\Big(\alpha (a,x)+(1-\alpha) (b,y)\Big)\leq \Big(\alpha g(a,z)^{-1}+(1-\alpha) g(b,z)^{-1}\Big)^{-1},
\end{align*}
where $z=\alpha x +(1-\alpha)y$. This means that
\begin{align}\label{qw}
g\big(\mathrm{\mathbf{A}}_\alpha(a,b),\mathrm{\mathbf{A}}_\alpha(x,y) \big)\leq
 \mathrm{\mathbf{H}}_\alpha \left\{ g(a,\mathrm{\mathbf{A}}_\alpha(x,y)),g(b,\mathrm{\mathbf{A}}_\beta(x,y))\right\}.
\end{align}
On the other hand we can write
\begin{align}\label{qn1}
f\left(\frac{a}{\mathrm{\mathbf{A}}_\alpha(x,y)}\right)=f\left(\frac{\frac{\alpha a}{yx}+\frac{(1-\alpha)a}{yx}}{\frac{\alpha }{y}+\frac{(1-\alpha)}{x}}\right)
=f\left(\beta\left(\frac{a}{x}\right)+(1-\beta)\left(\frac{a}{y}\right)\right),
\end{align}
where $\beta=\frac{\frac{\alpha}{y}}{\frac{\alpha}{y}+\frac{(1-\alpha)}{x}}=\frac{\alpha x}{\mathrm{\mathbf{A}}_\alpha(x,y)}$.  Since $f$ is convex,   \eqref{qn1} implies that
\begin{align*}
f\left(\frac{a}{\mathrm{\mathbf{A}}_\alpha(x,y)}\right)\leq \frac{\alpha x}{\mathrm{\mathbf{A}}_\alpha(x,y)}f\left(\frac{a}{x}\right)+
\frac{(1-\alpha)y}{\mathrm{\mathbf{A}}_\alpha(x,y)}f\left(\frac{a}{y}\right).
\end{align*}
Multiplying both sides by $\mathrm{\mathbf{A}}_\alpha(x,y)$ we get
\begin{align}\label{qw1}
g\big(a,\mathrm{\mathbf{A}}_\alpha(x,y)\big)\leq \mathrm{\mathbf{A}}_\alpha(g(a,x),g(a,y)).
\end{align}
Similarly
\begin{align}\label{qw2}
g\big(b,\mathrm{\mathbf{A}}_\alpha(x,y)\big)\leq \mathrm{\mathbf{A}}_\alpha(g(b,x),g(b,y)).
\end{align}
Since the Harmonic mean is monotone, it follows from  \eqref{qw1} and \eqref{qw2} that
\begin{align*}
g\big(\mathrm{\mathbf{A}}_\alpha(a,b),\mathrm{\mathbf{A}}_\alpha(x,y) \big)&\leq
 \mathrm{\mathbf{H}}_\alpha \left\{ g(a,\mathrm{\mathbf{A}}_\alpha(x,y)),g(b,\mathrm{\mathbf{A}}_\alpha(x,y))\right\}\qquad\mbox{(by \eqref{qw})}\\
 &\leq  \mathrm{\mathbf{H}}_\alpha \left\{ \mathrm{\mathbf{A}}_\alpha(g(a,x),g(a,y)),\mathrm{\mathbf{A}}_\alpha(g(b,x),g(b,y))\right\}
\end{align*}
which is the desired inequality \eqref{tm}. With $x=y$, this gives the AH-convexity of $g$ in the first coordinate and with $a=b$ this implies the  convexity of $g$ in the second  coordinate.\par

Conversely, if $g$ is AH-convexity in the first coordinate, then $f(t)=g(t,1)$ is an AH-convex function, too. This completes the proof of {\rm (i)}. \\
Next suppose that $f$ is \rm{AG}-convex.   For all $x,y,a\geq0$ and every $\alpha\in[0,1]$ we have
\begin{align*}
 g\big(\alpha a +(1-\alpha) b,x\big)&=xf\Big(\frac{1}{x}\big(\alpha a+(1-\alpha)b\big)\Big)\\
&\leq x\Big[f\big(\frac{a}{x}\big)^{\alpha}f\big(\frac{b}{x}\big)^{1-\alpha}\Big]\notag\\
&= \left[xf\big(\frac{a}{x}\big)\right]^{\alpha}\left[xf\big(\frac{b}{x}\big)\right]^{1-\alpha}\\
&= g(a,x)^{\alpha}  g(b,x)^{1-\alpha},\notag
\end{align*}
whence $g$ is \rm{AG}-convex in its first variable. Hence
\begin{align}\label{qa1}
g\big(\mathrm{\mathbf{A}}_\alpha(a,b),\mathrm{\mathbf{A}}_\alpha(x,y) \big)\leq
\mathrm{\mathbf{G}}_\alpha\left\{g(a,\mathrm{\mathbf{A}}_\alpha(x,y)),
g(b,\mathrm{\mathbf{A}}_\alpha(x,y))\right\}.
\end{align}
Furthermore,  taking into account the Arithmetic-Geometric means inequality, we know that $f$ is a convex  function so that   \eqref{qw1} and \eqref{qw2} hold. Regarding the monotonicity of the Geometric mean in its both variables we conclude from \eqref{qw1} and \eqref{qw2}  that
\begin{align}\label{qa2}
 \mathrm{\mathbf{G}}_\alpha\left\{g(a,\mathrm{\mathbf{A}}_\alpha(x,y)),
g(b,\mathrm{\mathbf{A}}_\alpha(x,y))\right\}
\leq
\mathrm{\mathbf{G}}_\alpha\left\{ \mathrm{\mathbf{A}}_\alpha(g(a,x),g(a,y)),\mathrm{\mathbf{A}}_\alpha(g(b,x),g(b,y))\right\}.
\end{align}
Combining \eqref{qa1} and \eqref{qa2} we infer \eqref{qaa}. Putting  $x=y$, the $\mathrm{AG}$-convexity of $g$ in first coordinate follows from \eqref{qaa} and with $a=b$, \eqref{qaa} gives the convexity of $g$ in the second coordinate. Conversely, if $g$ is $\mathrm{AG}$-convex    in its first coordinate, then $f(t)=g(t,1)$ is $\mathrm{AG}$-convex as well.\\
To prove {\rm (iii)}, let $f$ be a $\mathrm{GG}$-convex function. Then
\begin{align*}
  f\left(\frac{\mathrm{\mathbf{G}}_\alpha(a,b)}{\mathrm{\mathbf{G}}_\alpha(x,y)}\right)=
   f\left(a^\alpha b^{1-\alpha}x^{-\alpha}y^{\alpha-1}\right)&=
   f\left(\mathrm{\mathbf{G}}_\alpha\left(\frac{a}{x},\frac{b}{y}\right)\right)\leq  \mathrm{\mathbf{G}}_\alpha\left(f\left(\frac{a}{x}\right),f\left(\frac{b}{y}\right)\right).
\end{align*}
Hence
\begin{align*}
g\big(\mathrm{\mathbf{G}}_\alpha(a,b),\mathrm{\mathbf{G}}_\alpha(x,y) \big)&=
\mathrm{\mathbf{G}}_\alpha(x,y)
f\left(\frac{\mathrm{\mathbf{G}}_\alpha(a,b)}{\mathrm{\mathbf{G}}_\alpha(x,y)}\right)   \\
                        &\leq
   \mathrm{\mathbf{G}}_\alpha(x,y)\mathrm{\mathbf{G}}_\alpha\left(f\left(\frac{a}{x}\right),
f\left(\frac{b}{y}\right)\right)\\
   &=
\mathrm{\mathbf{G}}_\alpha\left\{ g(a,x),g(b,y))\right\}
\end{align*}
 as required. This proves  {\rm (iii)}. \\
Next suppose that $f$ is a   $\mathrm{HH}$-convex function. We write
\begin{align*}
  f\left(\frac{\mathrm{\mathbf{H}}_\alpha(a,b)}{\mathrm{\mathbf{H}}_\alpha(x,y)}\right)&=
   f\left(\left(\frac{\alpha a^{-1}+(1-\alpha)b^{-1}}{\alpha x^{-1}+(1-\alpha)y^{-1}}\right)^{-1}\right)\\
&=  f\left(\left(\frac{\alpha \frac{xy}{a}+ (1-\alpha)\frac{xy}{b}}{\alpha y+(1-\alpha)x}\right)^{-1}\right)\\
&=
 f\left(\left(\beta\ \frac{x}{a}+(1-\beta)\frac{y}{b}\right)^{-1}\right)
\end{align*}
in which we set $\beta=\frac{\alpha y}{\alpha y+(1-\alpha)x}$. Since $f$ is  $\mathrm{HH}$-convex, we obtain
\begin{align*}
  f\left(\frac{\mathrm{\mathbf{H}}_\alpha(a,b)}{\mathrm{\mathbf{H}}_\alpha(x,y)}\right)&=
   f\left(\mathrm{\mathbf{H}}_\beta\left(\frac{a}{x},\frac{b}{y}\right)\right)
\leq\mathrm{\mathbf{H}}_\beta\left(f\left(\frac{a}{x}\right),f\left(\frac{b}{y}\right)\right)
\end{align*}
so that
\begin{align}\label{qc1}
  g\left(\mathrm{\mathbf{H}}_\alpha(a,b),\mathrm{\mathbf{H}}_\alpha(x,y)\right)
 =\mathrm{\mathbf{H}}_\alpha(x,y)
f\left(\frac{\mathrm{\mathbf{H}}_\alpha(a,b)}{\mathrm{\mathbf{H}}_\alpha(x,y)}\right)
 \leq \mathrm{\mathbf{H}}_\alpha(x,y)
\mathrm{\mathbf{H}}_\beta\left(f\left(\frac{a}{x}\right),f\left(\frac{b}{y}\right)\right).
\end{align}
A simple calculation shows that
$$\mathrm{\mathbf{H}}_\alpha(x,y)
\mathrm{\mathbf{H}}_\beta\left(f\left(\frac{a}{x}\right),f\left(\frac{b}{y}\right)\right)=
\mathrm{\mathbf{H}}_\alpha\left\{g(a,x),g(b,y)\right\}. $$
Consequently, \eqref{qc} follows from \eqref{qc1}. Hence $g$ is jointly \rm{HH}-convex. Conversely, if $g$ is jointly \rm{HH}-convex, then $f(t)=g(t,1)$ is \rm{HH}-convex.
 This proves  {\rm (iv)}.\\
%%-----------------------------------------------------------------------------
Let $f$ be a \rm{GH}-convex function. For all $a,b,x\geq0$ and $\alpha\in[0,1]$ we have
\begin{align}\label{qv12}
g\big(\mathrm{\mathbf{G}}_\alpha(a,b),x \big)=
xf\left(\frac{\mathrm{\mathbf{G}}_\alpha(a,b)}{x}\right)&=
xf\left(\mathrm{\mathbf{G}}_\alpha\left(\frac{a}{x},\frac{b}{x}\right)\right)\nonumber\\
& \leq
x\mathrm{\mathbf{H}}_\alpha\left(f\left(\frac{a}{x}\right),f\left(\frac{b}{x}\right)\right)\nonumber\\
&=\mathrm{\mathbf{H}}_\alpha\left(g(a,x),g(b,x)\right),
\end{align}
whence $g$ is \rm{GH}-convex function in its first coordinate.
Furthermore,  we can write
\begin{align}\label{qv11}
 g\big(a,\mathrm{\mathbf{G}}_\alpha(x,y)\big)&=\mathrm{\mathbf{G}}_\alpha(x,y)
f\left(\frac{a}{\mathrm{\mathbf{G}}_\alpha(x,y)}\right)\nonumber\\
&=\mathrm{\mathbf{G}}_\alpha(x,y)
f\left(\mathrm{\mathbf{G}}_\alpha\left(\frac{a}{x},\frac{a}{y}\right)\right)\nonumber\\
&\leq\mathrm{\mathbf{G}}_\alpha(x,y) \mathrm{\mathbf{H}}_\alpha\left(f\left(\frac{a}{x}\right),f\left(\frac{a}{y}\right)\right)\nonumber\\
&\leq \mathrm{\mathbf{G}}_\alpha(x,y) \mathrm{\mathbf{G}}_\alpha\left(f\left(\frac{a}{x}\right),f\left(\frac{a}{y}\right)\right)=
\mathrm{\mathbf{G}}_\alpha\left(g(a,x),g(a,y)\right),
\end{align}
where the last inequality follows from the Harmonic-Geometric mean inequality. This ensures that $g$ is {\rm GG}-convex  in the second coordinate. Furthermore, combining \eqref{qv12} and \eqref{qv11} and using the monotonicity of the Harmonic mean, we reach \eqref{qv}. In addition, a similar argument as in \eqref{qv12} shows that  $g$ is jointly {\rm GG}-convex. Indeed,
\begin{align*}
g\big(\mathrm{\mathbf{G}}_\alpha(a,b),\mathrm{\mathbf{G}}_\alpha(x,y) \big)&=
\mathrm{\mathbf{G}}_\alpha(x,y)f\left(\frac{\mathrm{\mathbf{G}}_\alpha(a,b)}
{\mathrm{\mathbf{G}}_\alpha(x,y)}\right)\\
&= \mathrm{\mathbf{G}}_\alpha(x,y)
f\left(\mathrm{\mathbf{G}}_\alpha\left(\frac{a}{x},\frac{b}{y}\right)\right)\nonumber\\
& \leq
\mathrm{\mathbf{G}}_\alpha(x,y)
\mathrm{\mathbf{H}}_\alpha\left(f\left(\frac{a}{x}\right),f\left(\frac{b}{y}\right)\right)\nonumber\\
&\leq \mathrm{\mathbf{G}}_\alpha(x,y)
\mathrm{\mathbf{G}}_\alpha\left(f\left(\frac{a}{x}\right),f\left(\frac{b}{y}\right)\right)\nonumber\\
&=\mathrm{\mathbf{G}}_\alpha\left(g(a,x),g(b,y)\right),
\end{align*}
and so $g$ is jointly \rm{GG}-convex as we claimed. The converse follows similarly as previous parts.
\end{proof}
%%%%%%%%%%%%%%%%%%%%%%%%%%%%%%%%%%%%%%%%%%%%%%%%%%%%%%%%%%%%

We give some particular corollaries of of Theorem \ref{main1}  for some $f$-divergence functionals.
It is  easy to see that there are positive real numbers $c_1,c_2$ for which the function $f(t)=t\log t$ is \rm{AH}-convex on $(c_1,c_2)$. Theorem  \ref{main1} implies that the \textit{Kullback--Leibler distance} $KL(\mathbf{p},\mathbf{q})=\sum_{j=1}^{n}p_j \log\left(\frac{p_j}{q_j}\right)$ is \rm{AH}-convex on the first coordinate and  convex on the second  coordinate. As another example, the function $f(t)=t^r$ is \rm{AG}-convex on $(0,\infty)$ for all $r<0$. By Theorem \ref{main1}, the generated  divergence functional $ \rho_r(\mathbf{p},\mathbf{q})= \sum_{j=1}^{n}p_j^r q_j^{1-r}$ is \rm{AG}-convex on its first coordinate and  convex on its second  coordinate.
As another example, the function $f(t)=2(\sqrt{t}-1)^2$ is \rm{GG}-convex. Accordingly, the related divergence functional, which is the \textit{Hellinger distance} $H^2(\mathbf{p},\mathbf{q}):=2\sum_{j=1}^{n}\left(\sqrt{p_j}-\sqrt{q_j}\right)^2$ is jointly \rm{GG}-convex.

%%%%%%%%%%%%%%%%%%%%%%%%%%%%%%%%%%%%%%%%%%%%%%%%%%%%%%%%%%%%%%%%%%%%%%%%%%%%
Some applications of Theorem \ref{main1} will  be  given in the next section.
As we saw in Lemma \ref{j-all},   the \rm{MN}-convexity of $f$ produces variants of the Jensen inequality. Here, we study inequality \eqref{m1} in \rm{Theorem A}, when the core function $f$ enjoys \rm{MN}-convexity.

\begin{theorem} Let $\mathbf{a}=(a_1,\ldots,a_n)$ and $\mathbf{b}=(b_1,\ldots,b_n)$ be two $n$-tuples of positive real numbers and let $f$ be a positive real function on $(0,\infty)$. \\
\rm{(i)}\ If $f$ is an \rm{AH}-convex function, then
\begin{align}\label{eq.1}
g\left(\mathbf{\bar{a}},\mathbf{\bar{b}}\right)\leq \mathbf{\bar{a}}^2 I_{\frac{1}{f}}(\mathbf{a},\mathbf{b})^{-1}
\leq   I_f(\mathbf{a},\mathbf{b}).
\end{align}
\rm{(ii)}\ If $f$ is an \rm{AG}-convex function, then
\begin{align}\label{eq.12}
g\left(\mathbf{\bar{a}},\mathbf{\bar{b}}\right)\leq \mathbf{\bar{b}}\ \exp\left[\frac{1}{\mathbf{\bar{b}}}\ I_{\log f}(\mathbf{a},\mathbf{b})\right]
\leq   I_f(\mathbf{a},\mathbf{b}).
\end{align}
\rm{(iii)}\ If $f$ is  a \rm{HA}-convex function, then
\begin{align}\label{eq.12'}
g\left(\mathbf{\bar{a}},\mathbf{\bar{b}}\right)=\frac{\mathbf{\bar{b}}}{\mathbf{\bar{a}}}
g_{\varphi}\left(\mathbf{\bar{b}},\mathbf{\bar{a}}\right)=
\frac{\mathbf{\bar{b}}}{\mathbf{\bar{a}}}\ g_\phi\left(\mathbf{\bar{a}},\mathbf{\bar{b}}\right)\leq \frac{\mathbf{\bar{b}}}{\mathbf{\bar{a}}} I_\varphi(\mathbf{b},\mathbf{a})=\frac{\mathbf{\bar{b}}}{\mathbf{\bar{a}}}
I_\phi\left(\mathbf{{a}},\mathbf{{b}}\right),
\end{align}
where $\varphi(t)=f(1/t)$ and $\phi(t)=tf(t)$.\\
\rm{(iv)}\ If $f$ is an increasing \rm{GA}-convex function, then
\begin{align}\label{eq.123}
&g\left(\mathbf{\bar{a}},\mathbf{\bar{b}}\right)\leq g_{f o\exp}\left(\mathbf{\bar{a}},\mathbf{\bar{b}}\right)\leq I_{f o\exp}(\mathbf{a},\mathbf{b}).
\end{align}
\end{theorem}
\begin{proof}
Suppose that $\mathbf{a},\mathbf{b}$ are $n$-tuples of positive real numbers. For every $i=1,\ldots,n$, we  set  $\beta_i=\frac{b_i}{\sum_{k=1}^n b_k}$ so that $(\beta_1,\ldots,\beta_n)$ is a probability vector.  First assume that   $f$ is an \rm{AH}-convex function  for which we can write
\begin{align}\label{eq.2}
f\Big(\frac{\sum_{i=1}^na_i}{\sum_{i=1}^nb_i}\Big)&=f\Big(\frac{a_1}{\sum_{k=1}^nb_k}+\cdots+\frac{a_n}
{\sum_{k=1}^nb_k}\Big)\nonumber\\
&=f\Big(\frac{a_1}{b_1}\frac{b_1}{\sum_{k=1}^nb_k}+\cdots+\frac{a_n}{b_n}\frac{b_n}{\sum_{k=1}^nb_k}\Big)
\nonumber\\
&\leq \left(\sum_{i=1}^n\frac{b_i}{\sum_{k=1}^n b_k}\frac{1}{f(\frac{a_i}{b_i})}\right)^{-1}
\end{align}
where we use \eqref{jah} with $\beta_i=\frac{b_i}{\sum_{k=1}^n b_k}$.
Multiplying both sides of \eqref{eq.2}  with  $\sum_{k=1}^nb_k$  we get
\begin{align*}%\label{eq.3}
\big(\sum_{k=1}^nb_k\big)f\Big(\frac{\sum_{i=1}^na_i}{\sum_{i=1}^nb_i}\Big)&\leq \Big(\sum_{k=1}^nb_k\Big)\bigg( \sum_{i=1}^n\frac{b_i}{\sum_{k=1}^n b_k}\frac{1}{f(\frac{a_i}{b_i})}\bigg)^{-1}\\
&=\big(\sum_{k=1}^nb_k\big)^2\bigg( \sum_{i=1}^n{b_i}\frac{1}{f(\frac{a_i}{b_i})}\bigg)^{-1}\\
&=\Big(\sum_{k=1}^nb_k\Big)^2 I_{\frac{1}{f}}(\mathbf{a},\mathbf{b})^{-1},
\end{align*}
 which implies the first inequality in (\ref{eq.1}). To get the second inequality we use the convexity of the function $t\mapsto t^{-1}$.
\begin{align*}%\label{eq.3}
\Big(\sum_{k=1}^nb_k\Big)^2 I_{\frac{1}{f}}(\mathbf{a},\mathbf{b})^{-1}=\Big(\sum_{k=1}^nb_k\Big)\bigg( \sum_{i=1}^n\frac{b_i}{\sum_{k=1}^n b_k}\frac{1}{f(\frac{a_i}{b_i})}\bigg)^{-1}
\leq \sum_{i=1}^{n}b_i f\left(\frac{a_i}{b_i}\right)=I_f(\mathbf{a},\mathbf{b}).
\end{align*}
This completes the proof of \rm{(i)}. Next  assume that $f$ is an \rm{AG}-convex function. Then
\begin{align}\label{p-1}
  f\Big(\frac{\sum_{i=1}^na_i}{\sum_{i=1}^nb_i}\Big)&=
  f\left(\sum_{i=1}^{n}\beta_i \frac{a_i}{b_i}\right)\leq \prod_{i=1}^{n} f\left(\frac{a_i}{b_i}\right)^{\beta_i},
  \end{align}
 in which we use the same convex coefficients $\beta_i$ as in the proof of  \rm{(i)}.  Moreover,
\begin{align}\label{p-2}
\prod_{i=1}^{n} f\left(\frac{a_i}{b_i}\right)^{\beta_i}&=\prod_{i=1}^{n}\exp\left[\frac{b_i}{\sum_{k=1}^n b_k}\log f\left(\frac{a_i}{b_i}\right) \right]\nonumber\\
&=\exp\left[\sum_{i=1}^n\frac{b_i}{\sum_{k=1}^n b_k}\log f\left(\frac{a_i}{b_i}\right) \right]=\exp\left[\frac{1}{\mathbf{\bar{b}}}\ I_{\log f}(\mathbf{a},\mathbf{b})\right].
  \end{align}
The left inequality in \eqref{eq.12} follows from \eqref{p-1} and \eqref{p-2}. In addition, utilising the Arithmetic-Geometric means inequality   we reach
\begin{align}\label{p-3}
\prod_{i=1}^{n} f\left(\frac{a_i}{b_i}\right)^{\beta_i}\leq\sum_{i=1}^n\beta_i f\left(\frac{a_i}{b_i}\right)
=\frac{1}{\mathbf{\bar{b}}} I_f(\mathbf{a},\mathbf{b})
  \end{align}
  and the right inequality in \eqref{eq.12} is derived. This concludes \rm{(ii)}.

%----------------------------------------------------------
Now assume that $f$ is a   \rm{HA}-convex function. It is not hard to see that \cite{Drr} the functions $\varphi(t)=f(1/t)$ and $\phi(t)=tf(t)$ are  convex on proper domains so that \rm{Theorem A} gives
 $g_\varphi\left(\mathbf{\bar{a}},\mathbf{\bar{b}}\right)\leq I_\varphi(\mathbf{a},\mathbf{b})$ and
 $g_\phi\left(\mathbf{\bar{a}},\mathbf{\bar{b}}\right)\leq I_\phi(\mathbf{a},\mathbf{b})$.
We consider the convex coefficients $\alpha_i=\frac{a_i}{\sum_{k=1}^na_k}$ for $i=1,\ldots,n$ in such a way that
{\small\begin{align*}
 f\left(\frac{\sum_{i=1}^na_i}{\sum_{i=1}^nb_i}\right)=
 f\left(\left(\frac{\sum_{i=1}^nb_i}{\sum_{i=1}^na_i}\right)^{-1}\right)
 =f\left(\left(\sum_{i=1}^n\frac{b_i}{a_i}\frac{a_i}{\sum_{k=1}^na_k}\right)^{-1}\right)
 =f\left(\left(\sum_{i=1}^n\alpha_i\frac{b_i}{a_i}\right)^{-1}\right).
\end{align*}}
 By the \rm{HA}-convex of  $f$, this concludes that
\begin{align*}
 f\left(\frac{\sum_{i=1}^na_i}{\sum_{i=1}^nb_i}\right)\leq \sum_{i=1}^n\alpha_if\left(\frac{a_i}{b_i}\right)=
 \frac{1}{\sum_{k=1}^na_k}\sum_{i=1}^n a_i f\left(\frac{a_i}{b_i}\right).
\end{align*}
Multiplying both sides by $\sum_{k=1}^nb_k$ we reach
\begin{align*}
 g\left(\mathbf{\bar{a}},\mathbf{\bar{b}}\right)\leq  \frac{\mathbf{\bar{b}}}{\mathbf{\bar{a}}} I_\varphi(\mathbf{b},\mathbf{a}).
\end{align*}
On the other hand,
$$g\left(\mathbf{\bar{a}},\mathbf{\bar{b}}\right)=\mathbf{\bar{b}}
f\left(\frac{\mathbf{\bar{a}}}{\mathbf{\bar{b}}}\right)=
\frac{\mathbf{\bar{b}}}{\mathbf{\bar{a}}} \mathbf{\bar{a}}
\varphi\left(\frac{\mathbf{\bar{b}}}{\mathbf{\bar{a}}}\right)=\frac{\mathbf{\bar{b}}}{\mathbf{\bar{a}}}
g_\varphi\left(\mathbf{\bar{b}},\mathbf{\bar{a}}\right)=
\frac{\mathbf{\bar{b}}}{\mathbf{\bar{a}}}\ g_\phi\left(\mathbf{\bar{a}},\mathbf{\bar{b}}\right).
$$
Furthermore, we compute
\begin{align*}
I_\varphi\left(\mathbf{{b}},\mathbf{{a}}\right)=\sum_{i=1}^{n}a_i\varphi\left(\frac{b_i}{a_i}\right)=
\sum_{i=1}^{n}b_i\frac{a_i}{b_i}f\left(\frac{a_i}{b_i}\right)=
I_\phi\left(\mathbf{{a}},\mathbf{{b}}\right),
\end{align*}
so that we arrive at \rm{(iii)}.
  %--------------------------------------------------------------------------

  For proving \rm{(iv)}, first not that a function $f$ is \rm{GA}-convex if and only if the function $t\mapsto f(e^t)$ is convex, indeed, when proper domains are considered. So the Csisz\'{a}r inequality in \rm{Theorem A} implies the right  inequality of \eqref{eq.123}:
  \begin{align}\label{eq.112}
g_{f o\exp}\left(\mathbf{\bar{a}},\mathbf{\bar{b}}\right)\leq I_{f o\exp}(\mathbf{a},\mathbf{b}).
\end{align}
 When $f$ is increasing, we have $f o\exp\geq f$ on the positive half line. This ensure that the left inequality in \rm{(iv)} is valid.
\end{proof}

%%%%%%%%%%%%%%%%%%%%%%%%%%%%%%%%%%%%%%%%%%%%%%%%%%%%%%%%%%%%%%%%%%%%%%%%%%%%%%%%%%%%
\section{Matrix Jensen Inequality}
Let $\mathbb{M}_n$ denote the  algebra of $n\times n$  complex matrices and $I$ denote the identity matrix. It is known that (see for example \cite[Theorem 1.2]{FMPS})  an extension of the classical Jensen inequality holds as follows:
\begin{align}\label{o-j}
  f(\langle A\eta,\eta\rangle)\leq \langle f(A)\eta,\eta\rangle
\end{align}
for every continuous convex function $f:J\to\mathbb{R}$ and every Hermitian matrix $A\in\mathbb{M}_n$ with eigenvalues in $J$ and every unit vector $\eta\in\mathbb{C}^n$.     Our mean by $f(A)$ is the Hermitian matrix defined using the spectral decomposition of $A$. Indeed, if $A=\sum_{i=1}^n{\lambda_i}{P_i}$ is the spectral decomposition of the Hermitian matrix $A\in\mathbb{M}_n$, when $\lambda_i$'s are eigenvalues of $A$ and $P_i$'s are projections with $\sum_{i=1}^{n}P_i=I$, then $f(A)=\sum_{i=1}^nf(\lambda_i)P_i$, See \cite{FMPS}.

Lemma \ref{j-all} can be applied to derive variants of \eqref{o-j} for \rm{MN}-convex functions. See \cite{AD,HNM,Ki2,K-d,MK,RHM} and references therein for a collection of such inequalities.
\begin{proposition}\label{j-o-all}
  Let $f$ be a continuous real function. \\
{\rm (i)}  If $f$ is \rm{AH}-convex, then
\begin{align}\label{o-j-ah}
  f(\langle A\eta,\eta\rangle)\leq \langle f(A)^{-1}\eta,\eta\rangle^{-1},
\end{align}
{\rm (ii)}  If $f$ is \rm{AG}-convex, then
\begin{align}\label{o-ag}
  f(\langle A\eta,\eta\rangle)\leq \exp \langle \log f(A)\eta,\eta\rangle,
\end{align}
{\rm (iii)}  If $f$ is \rm{GA}-convex, then
\begin{align}\label{o-ga}
  f(\exp\langle \log A\eta,\eta\rangle)\leq   \langle f(A)\eta,\eta\rangle,
\end{align}
{\rm (iv)}  If $f$ is \rm{GG}-convex, then
\begin{align}\label{o-gg}
  f(\exp\langle \log A\eta,\eta\rangle)\leq   \exp\langle \log f(A)\eta,\eta\rangle,
\end{align}
{\rm (v)}  If $f$ is \rm{GH}-convex, then
\begin{align}\label{o-gh}
   f(\exp\langle \log A\eta,\eta\rangle)\leq    \langle f(A)^{-1}\eta,\eta\rangle^{-1},
\end{align}
{\rm (vi)}  If $f$ is \rm{HG}-convex, then
\begin{align}\label{o-hg}
  f\left(\left\langle  A^{-1}\eta,\eta\right\rangle^{-1}\right)\leq   \exp\langle \log f(A)\eta,\eta\rangle,
\end{align}
{\rm (vii)}  If $f$ is \rm{HH}-convex, then
\begin{align}\label{o-hh}
  f\left(\left\langle  A^{-1}\eta,\eta\right\rangle^{-1}\right)\leq  \langle   f(A)^{-1}\eta,\eta\rangle^{-1},
\end{align}
for every unit vector $\eta\in\mathbb{C}^n$ and every  Hermitian matrix $A\in\mathbb{M}_n$, whose eigenvalues are contained in the domain of $f$.
\end{proposition}
\begin{proof}
We only note that utilising    the    spectral decomposition of   $A$, the inner product terms in every part of proposition can be described by an scalar mean.   For example, if $A=\sum_{i=1}^n{\lambda_i}{P_i}$ is the spectral decomposition of   $A$, then $\sum_{i=1}^{n}  \langle P_i\eta,\eta\rangle=1$ and
\begin{align*}
\exp\langle \log A\eta,\eta\rangle&=\exp\left\langle \left(\sum_{i=1}^{n}\log\lambda_i P_i\right)\eta,\eta\right\rangle\\
&=\exp\left(\sum_{i=1}^{n}\langle P_i\eta,\eta\rangle\log\lambda_i\right)=
\prod_{i=1}^{n}{\lambda_i}^{\langle P_i\eta,\eta\rangle}=\rm{G}(\alpha;\Lambda),
\end{align*}
where $\alpha=(\langle P_1\eta,\eta\rangle,\ldots,\langle P_n\eta,\eta\rangle)$ is a weight vector and $\Lambda=(\lambda_1,\ldots,\lambda_n)$.
\end{proof}
 %%%%%%%%%%%%%%%%%%%%%%%%%%%%%%%%%%%%%%%%%%%%%%%%%%%%%%%%%%%%%%%%%%%%%%%%%%%%%%%%%%%
Let $A\in\mathbb{M}_n$ and $B\in\mathbb{M}_m$ be Hermitian matrices with spectral decompositions $A=\sum_{i=1}^n{\lambda_i}{P_i}$ and  $B=\sum_{i=1}^m{\mu_i}{Q_i}$.
When $f$ is a   two variable real function defined on     $J_1\times J_2\subseteq\mathbb{R}^2$, then we can define a Hermitian  matrix $f(A,B)$ as  $$f(A,B)=\sum_{i=1}^{n}\sum_{j=1}^{m}f(\lambda_i,\mu_j)P_i\otimes Q_j$$
and so $f$ becomes a matrix function of two variables from $\mathbb{M}_n\times\mathbb{M}_m$ to $\mathbb{M}_{nm}$. It has been shown in \cite{Ma-Au} that if $f$ is a separately convex function on     $J_1\times J_2\subseteq\mathbb{R}^2$, then
 \begin{align}\label{2-vari}
  f(\langle A\eta,\eta\rangle,\langle B\zeta,\zeta\rangle)\leq
\langle f(A,B)\eta\otimes\zeta,\eta\otimes\zeta\rangle
\end{align}
for all unit vectors  $\eta\in\mathbb{C}^n$ and $\zeta\in\mathbb{C}^m$ and all  Hermitian matrices $A\in\mathbb{M}_n$ and $B\in\mathbb{M}_m$.

As it was shown in Theorem \ref{main1},   the type of convexity of the core function $f$ affects on the convexity of  perspective function $g$.  In the rest of this section, we are going to establish matrix Jensen inequality \eqref{2-vari} for the perspective functions   in the case where $f$ is a \rm{MN}-convex function.

\begin{theorem}\label{main3}
Let  $h$ be a real two-variable function on     $J_1\times J_2\subseteq\mathbb{R}^2$.\\
{\rm (i)}  If $h$ is separately \rm{HH}-convex, then
\begin{align}\label{2-v-hh}
  h\left(\langle A^{-1}\eta,\eta\rangle^{-1},\langle B^{-1}\zeta,\zeta\rangle^{-1}\right)\leq
\left\langle h(A,B)^{-1}\eta\otimes\zeta,\eta\otimes\zeta\right\rangle^{-1};
\end{align}
{\rm (ii)}  If $h$ is separately \rm{GG}-convex, then
\begin{align}\label{2-v-gg}
  h\left(\exp\langle \log A\eta,\eta\rangle,\exp\langle \log B\zeta,\zeta\rangle\right)\leq
\exp\left\langle \log h(A,B)\eta\otimes\zeta,\eta\otimes\zeta\right\rangle;
\end{align}
for all unit vectors  $\eta\in\mathbb{C}^n$ and $\zeta\in\mathbb{C}^m$ and all  Hermitian matrices $A\in\mathbb{M}_n$ and $B\in\mathbb{M}_m$.
\end{theorem}
\begin{proof}
  Suppose that $A=\sum_{i=1}^n{\lambda_i}{P_i}$ and  $B=\sum_{i=1}^m{\mu_i}{Q_i}$ are spectral decompositions of Hermitian matrices $A$ and $B$. Assume that $\eta\in\mathbb{C}^n$ and $\zeta\in\mathbb{C}^m$ are unit vectors so that $\sum_{i=1}^{n}\langle P_i\eta,\eta\rangle=1=\sum_{j=1}^{m}\langle Q_j\zeta,\zeta\rangle$. Then
{\small\begin{align}\label{hh2-1}
 h\left(\langle A^{-1}\eta,\eta\rangle^{-1},\langle B^{-1}\zeta,\zeta\rangle^{-1}\right)
&= h\left(\left(\sum_{i=1}^n{\lambda_i}^{-1}\langle P_i\eta,\eta\rangle\right)^{-1},b\right)\nonumber\\
&\leq \left(\sum_{i=1}^n\langle{P_i}\eta,\eta\rangle h(\lambda_i,b)^{-1}\right)^{-1},
\end{align}}
where $b=\langle B^{-1}\zeta,\zeta\rangle^{-1}$ and the inequality follows from the \rm{HH}-convexity of $h$ in the first variable and \eqref{jhh} of Lemma \ref{j-all}. Furthermore, for every $i=1,\ldots,n$, the \rm{HH}-convexity of $h$ in the second variable gives
{\small\begin{align}\label{hh2-2}
h(\lambda_i,b)= h\left(\lambda_i,\left(\sum_{j=1}^m{\mu_j}^{-1}\langle{Q_j}\zeta,\zeta\rangle\right)^{-1}\right)
&\leq \left(\sum_{j=1}^m h(\lambda_i,\mu_j)^{-1}\langle{Q_j}\zeta,\zeta\rangle\right)^{-1}.
\end{align}}
It follows from \eqref{hh2-1} and \eqref{hh2-2} that
{\small\begin{align*}
 h\left(\langle A^{-1}\eta,\eta\rangle^{-1},\langle B^{-1}\zeta,\zeta\rangle^{-1}\right)
&\leq  \left(\sum_{i=1}^n\sum_{j=1}^m\langle{P_i}\eta,\eta\rangle\langle{Q_j}\zeta,\zeta\rangle
h(\lambda_i,\mu_j)^{-1}\right)^{-1}\\
&=\left(\sum_{i=1}^n\sum_{j=1}^m  h(\lambda_i,\mu_j)^{-1} \left\langle (P_i\otimes Q_j)\eta\otimes\zeta,\eta\otimes\zeta\right\rangle\right)^{-1}\\
&=\left\langle h(A,B)^{-1}\eta\otimes\zeta,\eta\otimes\zeta\right\rangle^{-1}
\end{align*}}
and we obtain \eqref{2-v-hh}. Next suppose that $h$ is separately \rm{GG}-convex and suppose that $A$ and $B$ are Hermitian matrices with the same spectral decompositions as in the first part.  Utilising Lemma \ref{j-all} we have
{\small\begin{align}\label{gg2-1}
  h\left(\exp\langle \log A\eta,\eta\rangle,\exp\langle \log B \zeta,\zeta\rangle\right)
= h\left(\prod_{i=1}^{n}{\lambda_i}^{\langle P_i\eta,\eta\rangle},b\right)
\leq \prod_{i=1}^{n}h({\lambda_i},b)^{\langle P_i\eta,\eta\rangle},
\end{align}}
in which $b=\exp\langle \log B \zeta,\zeta\rangle=\prod_{j=1}^{m}{\mu_j}^{\langle Q_j \zeta,\zeta\rangle}$.  Moreover,  another  use of Lemma \ref{j-all} regarding the \rm{GG}-convexity of $h$  in the second variable gives
{\small\begin{align}\label{gg2-2}
 h({\lambda_i},b)=h\left(\lambda_i,\prod_{j=1}^{m}{\mu_j}^{\langle Q_j \zeta,\zeta\rangle}\right) \leq \prod_{j=1}^{m}h(\lambda_i,\mu_j)^{\langle Q_j \zeta,\zeta\rangle}
\end{align}}
for every $i=1,\ldots,n$. From \eqref{gg2-1} and \eqref{gg2-2} we obtain
{\small\begin{align*}
  h\left(\exp\langle \log A\eta,\eta\rangle,\exp\langle \log B \zeta,\zeta\rangle\right)
&\leq \prod_{i=1}^{n} \prod_{j=1}^{m}h(\lambda_i,\mu_j)^{\langle P_i\eta,\eta\rangle\langle Q_j \zeta,\zeta\rangle}\\
&=\exp\left\langle \log h(A,B)\eta\otimes\zeta,\eta\otimes\zeta\right\rangle
\end{align*}}
and we are done.
\end{proof}
%%%%%%%%%%%%%%%%%%%%%%%%%%%%%%%%%%%%%%%%%%%%%%%%%%%%%%%%%%%%%%%%%%%%%%%%%%%%%%%%%%%%%%%%
\begin{remark}
Let us give some applications of Theorem  \ref{main1} and  Theorem \ref{main3}  for perspective functions.
We show  in Theorem \ref{main1} that if $f$ is \rm{HH}-convex, then the associated perspective function $g$ is \rm{HH}-convex in its both variables and so   \eqref{2-v-hh} holds by  Theorem \ref{main3}. For example, the function $f(t)=t^r$ is  \rm{HH}-convex for every $r\in[0,1]$ and so $g(t,s)=sf(t/s)=s^{1-r}t^r$ is \rm{HH}-convex in its both variables. Note that in this particular example we have $g(A,B)=A^r\otimes B^{1-r}$.
Now \eqref{2-v-hh}  implies that
\begin{align*}
 \langle A^{-1}\eta,\eta\rangle^{-r}\langle B^{-1}\zeta,\zeta\rangle^{r-1}\leq \left\langle(A^r\otimes B^{1-r})\eta\otimes\zeta,\eta\otimes\zeta\right\rangle.
\end{align*}
Note that because the function  $g(t,s)$ in this example can be decomposed as $g(t,s)=g_1(t)g_2(s)$,  the above inequality follows directly from \eqref{o-j}.

\end{remark}
%%%%%%%%%%%%%%%%%%%%%%%%%%%%%%%%%%%%%%%%%%%%%%%%%%%%%%%%%%%%%%%%%%%%%%%%%%%%%%%%%%%%%%%%%%

\begin{remark}
If the type of convexity in first coordinate of  a two variable function $h$  differs from its second coordinate, it is also possible to present Theorem \ref{main3}. For example, assume that  the function $h$,  defined on $J_1\times J_2\subseteq\mathbb{R}^2$, is \rm{AH}-convex in its first coordinate and convex in the second coordinate. Then
\begin{align}\label{l1}
  h\left(\langle A\eta,\eta\rangle,\langle B\zeta,\zeta\rangle\right)\leq \left\langle h(A,\langle B\zeta,\zeta\rangle)^{-1}\eta,\eta\right\rangle^{-1},
\end{align}
 where we use \eqref{o-j-ah} in Proposition \ref{j-o-all} for the \rm{AH}-convex function $h_1(t)=h(t,\langle B\zeta,\zeta\rangle)$.  Note that $h(A,\langle B\zeta,\zeta\rangle)$ is a Hermitian matrix in $\mathbb{M}_n$ defined by
$$h(A,\langle B\zeta,\zeta\rangle)=\sum_{i=1}^{n}h(\lambda_i,\langle B\zeta,\zeta\rangle)P_i,$$
in which we use the spectral decomposition of $A$ as before.
Since $h$ is convex on the second coordinate,  for every $i=1,\ldots, n$ we have
\begin{align*}
h\left(\lambda_i,\langle B\zeta,\zeta\rangle\right)\leq\sum_{j=1}^{m}h(\lambda_i,\mu_j)\langle Q_j\zeta,\zeta\rangle.
\end{align*}
Accordingly,
\begin{align*}
h(A,\langle B\zeta,\zeta\rangle)\leq\sum_{i=1}^{n}
\sum_{j=1}^{m}h(\lambda_i,\mu_j)\langle Q_j\zeta,\zeta\rangle P_i,
\end{align*}
whence
{\small\begin{align}\label{l2}
\langle h(A,\langle B\zeta,\zeta\rangle)\eta,\eta\rangle\leq\sum_{i=1}^{n}
\sum_{j=1}^{m}h(\lambda_i,\mu_j)\langle Q_j\zeta,\zeta\rangle \langle P_i\eta,\eta\rangle=\langle h(A,B)\eta\otimes\zeta,\eta\otimes\zeta\rangle.
\end{align}}
Now we obtain from \eqref{l1} and \eqref{l2} that
\begin{align*}
   h\left(\langle A\eta,\eta\rangle,\langle B\zeta,\zeta\rangle\right)&\leq \left\langle h(A,\langle B\zeta,\zeta\rangle)^{-1}\eta,\eta\right\rangle^{-1}\\
&\leq \langle h(A,\langle B\zeta,\zeta\rangle)\eta,\eta\rangle\\
&\leq\langle h(A,B)\eta\otimes\zeta,\eta\otimes\zeta\rangle.
\end{align*}
Similarly, it can be shown that
\begin{align*}
   h\left(\langle A\eta,\eta\rangle,\langle B\zeta,\zeta\rangle\right)\leq \left\langle h(\langle A\eta,\eta\rangle,B)\zeta,\zeta\right\rangle\leq\langle h(A,B)\eta\otimes\zeta,\eta\otimes\zeta\rangle.
\end{align*}
\end{remark}

We show in Theorem \ref{main1} that if $f$ is \rm{AH}-convex, then its perspective function $g$ is
\rm{AH}-convex in first coordinate and convex in the second coordinate. Hence, the two last series of inequalities holds true for the   perspective function of every \rm{AH}-convex.

%!!!!!!!!!!!!!!!!!!!!!!!!!!!!!!!!!!!!!!!!!!!!!!!!!!!!!!!!!!!!!!!!!!!!!!!!!!!!!!!!!!!!!

\end{document}